\documentclass{article}
\usepackage{graphicx} \usepackage{paper}

\newcommand{\haim}[1]{\textcolor{red}{[Haim: #1]}}

\newcommand{\AAA}{\mathcal{A}}

\author{
Haim Kaplan\textsuperscript{$\dagger$,*} 
\and 
\hspace{-8px}Shay Sapir\textsuperscript{$\ddagger$,*} 
\and 
\hspace{-8px}Uri Stemmer\textsuperscript{$\dagger$,*}
}

\title{Load Balancing under Adaptive Bin Deletions}

\date{}

\theoremstyle{plain}

\newtheorem{remark}[definition]{Remark}  

\begin{document}

\maketitle

\begingroup
\renewcommand{\thefootnote}{} 
\footnotetext{
  \hspace{-7px}
  \textsuperscript{*}Google Research \quad
  \textsuperscript{$\dagger$}Tel Aviv University \quad
  \textsuperscript{$\ddagger$}Weizmann Institute of Science
}
\endgroup

\begin{abstract}
We analyze a balls-and-bins game against an adaptive adversary that sequentially deletes bins. Starting with $n$ balls distributed across $n$ bins, the adversary deletes a bin in each step, forcing the algorithm to redistribute its balls to surviving bins. We prove that after $n/2$ rounds, uniform random redistribution yields optimal $O(n)$ recourse and $O(\frac{\log n}{\log \log n})$ maximum load. Furthermore, we show that applying the ``power of two choices'' reduces the maximum load to $O(\log \log n)$ while maintaining linear recourse. 

We also consider a variation of this game where the balls from the deleted bin are partitioned evenly among $d \ll n$ random bins rather than being redistributed independently. We demonstrate that keeping the balls together ($d=1$), which gives small maximum load and recourse against an oblivious adversary, fails against an adaptive adversary. Nevertheless, we show that splitting the balls into just two groups ($d=2$) is sufficient to recover linear recourse and efficient load balancing in the adaptive setting.
\end{abstract}

\section{Introduction}

Balls-and-bins processes serve as a fundamental abstraction for studying load balancing problems, with diverse applications across numerous domains~\cite{AZAR199473,Karger_consistent97,RaabS98,MitzenmacherUpfal,Wieder07,flajolet2007hyperloglog,weinberger2009feature}. In the simplest  version of the problem, $n$ balls are assigned to $n$ bins. Here it is well established that distributing every ball independently and uniformly over the bins yields a maximum load of $O(\frac{\log n}{\log \log n})$ \cite{Gonnet81}. Furthermore, utilizing the ``power of two choices'', where balls are placed sequentially in the least loaded of two randomly chosen bins, dramatically reduces the maximum load to $O(\log \log n)$ \cite{AzarBKU99}.

Beyond these classical settings, the landscape of balls-and-bins problems is vast. Variations include settings with weighted balls~\cite{BERENBRINK_FHM08,PeresTW10}, heterogeneous bin capacities~\cite{BERENBRINK20142065,Wieder07}, and graph-based constraints where balls can probe only specific subsets of bins~\cite{Kenthapadi_Panigrahy06,Brighten08}.  
Berenbrink et al.~\cite{BCSV00_2choice_heavy,BCSV06_2choice_heavy-sicomp} extended the power of two choices from the {\em lightly-loaded} case studied by \cite{AzarBKU99}, where the number of balls $m$ is equal to the number of bins $n$, to the {\em heavily-loaded} case, where $m\gg n$. Remarkably, they showed that the additive gap between the maximum load and the average load remains the same as in the lightly-loaded case.
Los and Sauerwald \cite{LosS22} analyzed the power of two choices in an ``incomplete information'' setting, investigating scenarios where algorithms must rely on load estimates rather than exact values. V{\"o}cking \cite{vocking2003asymmetry} considered multiple-choice settings (rather than just two choices), where balls are placed in the least loaded of $d$ selected bins, demonstrating that a nonuniform selection of the $d$ locations leads to better load balancing.

Of particular relevance to our work are {\em dynamic} variations of balls-and-bins, where the goal is to maintain a balanced allocation over time as balls or bins arrive and depart. This line of work was initiated by \cite{AzarBKU99}, who considered a model where balls are added and deleted {\em at random}. They showed that their ``power of two choices'' extends to this dynamic setting, maintaining a maximum load of $O(\log \log n)$ in the lightly-loaded case. Subsequent works \cite{cole1998balls,vocking2003asymmetry,BansalK22} generalized this to the heavily-loaded case. Furthermore, these works allowed the sequence of operations to be determined by an {\em oblivious adversary}, that is, an adversary that fixes the full sequence of ball insertions and deletions before the process begins. Note that the random case studied by \cite{AzarBKU99} is a special instance of an oblivious adversary.

\paragraph*{Adaptive adversary.} Unlike an oblivious adversary, who commits to the sequence of operations before the game begins, an adaptive adversary chooses its actions sequentially, potentially basing its decisions on the entire history of the execution. Designing algorithms and data structures that cope with adaptive adversaries has been the focus of much recent work across various sub-fields of theoretical computer science, including streaming algorithms, dynamic graph algorithms, and learning theory
\cite{Ben-EliezerY20,BJWY22,HassidimKMMS22,max-flow-via2022,BeimelKMNSS22,BernsteinBGNSS022,AlonBDMNY21,HaeuplerKRS26}.
However, in the context of dynamic balls-and-bins, to our knowledge, only two works have addressed general adaptive adversaries: Bender et al.~\cite{BenderCCFJT29} and Fine et al.~\cite{FineKS25}.

Another work by Becchetti et al. \cite{BecchettiCNPP19}, considered a repeated balls-and-bins game, which can be viewed as \emph{specific adaptive strategy}. In this game, $n$ balls are arbitrarily assigned to $n$ bins, and in each round, one ball is extracted from each non-empty bin, and re-assigned to a uniformly random bin. Becchetti et al. \cite{BecchettiCNPP19} proved that this game admits a self-stabilizing property: after $O(n)$ rounds, the maximum load is $O(\log n)$, and remains $O(\log n)$ for $\poly(n)$ rounds, with high probability.\footnote{High probability means with probability at least $1/\poly
(n)$.}
However, the focus of Becchetti et al. \cite{BecchettiCNPP19} is on that specific adaptive strategy, hence it does not bound general adaptive adversaries. 

Bender et al.~\cite{BenderCCFJT29} considered an adaptive {\em ball-recycling} game where $m$ balls are initially distributed across $n$ bins. At each time step, an adversary adaptively selects a bin, removes all its balls, and re-throws them at random. The goal of the adversary is to maximize the expected number of balls re-thrown per step in the stationary distribution (when it exists), a metric we refer to as the {\em long-run average recourse}. Bender et al.~\cite{BenderCCFJT29} established nearly tight bounds for this setting: there exists an adversary that achieves a long-run average recourse of $2m/(n + 1)$, while no adversary can achieve long-run average recourse greater than $2m/n + 1$. However, since the work of \cite{BenderCCFJT29} focuses on the stationary distribution, it does not yield strong bounds on the achievable recourse in finite-horizon games.

This gap was addressed by Fine et al.~\cite{FineKS25}, who considered a variant of this game in which bins are {\em deleted} rather than {\em recycled}. This implies that the game can last at most $n$ steps. This is formulated as a two-player game between an (adaptive) adversary and a randomized algorithm, defined as follows.

\begin{definition}[The Bin-Deletion Game]\label{def:BDGame}
Initially, the algorithm distributes $n$ balls into $n$ bins according to its initialization strategy (in our setting, this will either be the uniform distribution or using the power of two choices). Then the game proceeds for $T\leq n$ rounds, where in each round:
\begin{enumerate}
\item The adversary inspects the current assignment and selects a bin to delete.
\item The algorithm redistributes the balls from the deleted bin among the surviving bins.
\end{enumerate}
The algorithm has two primary objectives: to minimize the total number of balls it re-distributes, referred to as the (total) \emph{recourse}, and to minimize the \emph{maximum load} of any single bin throughout the game.
\end{definition}

\begin{remark} 
As is standard in the literature on balls-and-bins, we focus on ``lightweight'' strategies for the algorithm, where the destination of a ball is chosen either independently of the current system state or by probing only a constant number of randomly chosen bins. Such strategies are crucial in distributed settings where querying the load of every bin may be slow or expensive.
\end{remark}

Fine et al.~\cite{FineKS25} considered this game in the case where $T=n-1$ (i.e., the game continues until only one bin remains). They showed that no adversary can force more than $O(n\log n)$ total recourse against the simple algorithm that redistributes balls from deleted bins uniformly at random. This is optimal for $T=n-1$, because at any step with $r$ remaining bins, the adversary can force a recourse of at least $n/r$ (the average load). Thus, even over just the last $n/2$ steps, any algorithm must incur a total recourse of:
$$\sum_{r=1}^{n/2} \frac{n}{r} = n \sum_{r=1}^{n/2} \frac{1}{r} = \Theta(n \log n).$$
Note that this argument fails when the number of deletions $T$ is smaller than $n/2$. The analysis of recourse and maximum load when the adversary removes only a subset of the bins was left open by Fine et al.~\cite{FineKS25}. (The maximum load is irrelevant when $T=n-1$, as it is trivially $n$ when only one bin remains.)

\subsection{Our results}
We analyze the recourse and the maximum load of several algorithmic strategies for the adaptive bin deletion game (Definition~\ref{def:BDGame}). First, we resolve the behavior of this game for partial deletions. For simplicity of presentation, we only discuss here $T=\tfrac{n}{2}$ rounds, and present the extension to $T>\tfrac{n}{2}$ in \Cref{sec:uniform}. We establish that after $\tfrac{n}{2}$ rounds, independent uniform redistribution yields optimal $O(n)$ recourse and a maximum load of $O(\tfrac{\log n}{\log\log n})$, matching the classical static bounds. Moreover, we show that the power of two choices carries over to this adaptive setting, reducing the maximum load to $O(\log \log n)$ and recovering the exponential improvement seen in the static case. Specifically,

\begin{theorem}[Informal]\label{thm:uniform_informal}
    In the adaptive bin-deletion game, after $T=\tfrac{n}{2}$ rounds, with high probability,
    \begin{enumerate}
        \item If balls are distributed independently and uniformly at random, the recourse is $O(n)$ and the maximum load is $O(\tfrac{\log n}{\log\log n})$, and
        \item If balls are distributed using the two choices allocation scheme, the recourse is $O(n)$ and the maximum load is $O(\log\log n)$.
    \end{enumerate}
\end{theorem}

\Cref{thm:uniform_informal} assumes that a deleted bin can send every one of its balls independently into a different random (remaining) bin.
However, this might sometimes be infeasible. For example, in a communication network, this might require a server to connect to a prohibitively large number of machines.
Therefore, we seek to reduce the number of target bins that a deleted bin connects to. We consider a restrictive algorithmic strategy called {\em $d$-split}: when a bin is deleted, we partition its balls among $d$ randomly chosen bins.
This strategy was suggested in \cite{FineKS25}, who asked what is the recourse of this scheme, particularly when $d=2$.
In \Cref{sec:2_split}, we answer this question when the game is played for $T=n/2$ rounds, showing that this $d$-split strategy suffices for guaranteeing linear recourse and nearly-optimal maximum load, even for $d=2$.
The analysis of this variant is significantly more challenging, and it may provide insights for handling adaptive inputs in other settings. Extending our result to $T=(1-o(1))n$ rounds remains open. Note that for this variant, our bound on the maximum load holds only with constant probability. 

\begin{theorem}\label{thm:2_split_infromal}
    In $2$-split against an adaptive adversary, after $\tfrac{n}{2}$ rounds, the recourse is $O(n)$ with high probability, and the maximum load is $\poly(\log n)$, with constant probability.
\end{theorem}

As noted by \cite{FineKS25}, restricting this variant further to $d=1$, which forces a deleted bin to send all its balls into {\em one} randomly chosen bin, is disastrous against an adaptive adversary: The adversary can always delete the heaviest bin, causing a snowball effect where a heavy packet of balls repeatedly merges with other bins. Already after $n/2$ rounds, this results in maximum load $\Omega(n)$ and recourse $\Omega(n^2)$.

Somewhat surprisingly, we show that this restrictive variant (with $d=1$) is actually a solid strategy against an {\em oblivious} adversary that chooses the sequence of deletions in advance. This underscores the power of the adaptive adversary and sets it apart from the oblivious adversary in the context of dynamic balls-and-bins. Specifically, in \Cref{sec:no_split} we prove the following theorem:

\begin{theorem}\label{thm:1_split_oblivious_intro}
    In $1$-split against an oblivious adversary, after $\tfrac{n}{2}$ rounds, the recourse is $O(n\log n)$ and the maximum load is $O(\log n)$, with high probability.
\end{theorem}

To grasp the intuition behind this result, consider a random \emph{forest} graph on $n$ vertices constructed as follows: for each step $i$ from $1$ to $n/2$, vertex $i$ uniformly samples a parent $j \in \{i+1, \dots, n\}$ and becomes a child of $j$. This process yields a forest with $n/2$ roots (the vertices indexed $n/2+1$ through $n$).

Critically, the recourse in our oblivious $d=1$ game is exactly equal to the sum of the depths of all vertices in this forest, and the maximum load is the size of the largest tree. This equivalence holds because, without loss of generality, an oblivious adversary deletes bins $1$ through $n/2$ sequentially. When bin $i$ is deleted, it merges into a surviving bin $j > i$, creating the directed edge $(i, j)$. Consequently, a ball starting at bin $u$ incurs a movement cost exactly equal to the length of the path from $u$ to a root; and the load of bin $i$ is the size of the tree rooted at $i$.
We provide a direct analysis showing that, for such a random forest, the total sum of depths is $O(n\log n)$, and the maximum tree size is $O(\log n)$ with high probability.

\subsection{Technical overview}
Our proofs use coupling arguments to show stochastic domination, formalized as follows.

\begin{definition}[First-order stochastic domination]\label{def:frst_ordr_stochastic_domination}
    Random variable $A$ has first-order stochastic dominance over random variable $B$, denoted $A\succeq B$ if, for all $x\in\R$,
    \[
    \Pr(A\geq x)\geq \Pr(B\geq x).
    \]
\end{definition}
\begin{lemma}[Theorem 4.2.3 of \cite{roch_mdp_2024}.]\label{lem:stoch_dominat_equivalent_coupling}
    For random variables $A,B$, we have $A\succeq B$ if and only if 
    there exist random variables $A',B'$, such that $A'$ has the same distribution as $A$, and $B'$ has the same distribution as $B$, and $A'\geq B'$. The random variables $A',B'$ are called a monotone coupling.
\end{lemma}
For brevity, from now on we omit the term ``first-order'' when referring to \Cref{def:frst_ordr_stochastic_domination}, and omit the word ``monotone'' when referring to \Cref{lem:stoch_dominat_equivalent_coupling}.
For concreteness, suppose that the number of rounds is $n/2$, for all variants.

\paragraph*{Uniform game.}
Consider the uniform algorithm, that for each ball, draws a bin uniformly and independently at random from the surviving bins, and assigns the ball to the drawn bin.
We show that the recourse and maximum load in this adaptive game are stochastically dominated by the recourse and maximum load in a related \emph{non-adaptive game}, and then bound the recourse and maximum load of the  non-adaptive game using standard tools.
The non-adaptive game is the following: throw balls independently and uniformly at random to $n$ bins, until every subset of $n/2$ bins contains at least $n$ balls. It is not difficult to show that with high probability, the number of balls in this game is at most $10n$, and the bound on the maximum load follows by classic results, for a standard balls and bins problem with $m=10n$.

The coupling argument, which is proven in \Cref{sec:uniform}, is as follows.
Whenever we want to throw a ball in the adaptive game, we do the following. In the non-adaptive game, throw a ball uniformly at random over the $n$ bins. If the ball lands in a bin of the same index as a surviving bin in the adaptive game, place the ball in the adaptive game at that corresponding bin. Otherwise, repeat this and throw another ball in the non-adaptive game.
We obtain that the configuration, at the bins which correspond to surviving bins in the adaptive game, is the same in both games.
Moreover, the non-adaptive game may only end after round $n/2$ of the adaptive game, since at round $n/2-1$, there are $n/2+1$ surviving bins, which contain $n$ balls, and so there is a set of $n/2$ bins with strictly less than $n$ balls.
Thus, the configuration in the non-adaptive game is point-wise larger than the configuration in the adaptive game, and the proof follows.

\begin{remark}
    It is necessary that the algorithm's coins are hidden from the adversary, at least in order to bound the maximum load.
    To see this, we first define a model where the algorithm's random coins are set in advance and revealed to the adversary.
    Suppose that every bin $b$ holds a sequence $\sigma_b$ of independent random variables drawn from the uniform distribution over $[n]$.
    When bin $b$ is deleted, it sends its balls to surviving bins in $\sigma_b$, in the order of $\sigma_b$ (this is basically rejection sampling).

    Given the sequences of all bins, the adversary can do the following attack.
    Look at the first $\sqrt{n}$ random variables of the first $\sqrt{n}$ bins. With constant probability, there exists $i\in[\sqrt{n}]$ such that $n$ appears in the first $\sqrt{n}$ random variables of $i$. Delete the other $\sqrt{n}$ bins indexed by random variables in $\sigma_i$, and then delete $i$. 
    This yields an additional $O(1)$ expected load on bin $n$.
    Repeat $O(\sqrt{n})$ times, resulting in expected load $\Omega(\sqrt{n})$ in bin $n$ after $n/2$ rounds.
\end{remark}

\paragraph*{Two choices.}
In the adaptive two choices paradigm, each ball draws two bins independently and uniformly over the surviving bins, and picks the least loaded bin.
Our proof for this setting is similar to that of the uniform game, except that the non-adaptive game we reduce to (via coupling) is different. Specifically, we consider the following non-adaptive variant of the two choices game (without an adversary): Start with $n$ empty bins. Then, iteratively pick {\em pairs of indices} $(i,j)$ uniformly from $[n]\times [n]$, and place a ball in the least loaded among bins $i,j$. The game proceeds until for every subset $S\subseteq[n]$ of size $|S|=n/2$ we have that
$S\times S \subset [n]\times [n]$ contains at least $n$ pairs. We show that the recourse and maximum load in the {\em adaptive} two choices game are stochastically dominated by the number of rounds and the maximum load in the \emph{non-adaptive} two choices game, and then bound these parameters in the non-adaptive game.

\paragraph*{$2$-split, bounding the recourse.}
Recall that with the 2-split strategy, whenever a bin is deleted, its balls are split into two packets, and each packet is independently thrown to a uniformly random surviving bin. \Cref{thm:2_split_infromal} states that after $\tfrac{n}{2}$ rounds, the recourse is $O(n)$ and the maximum load is $\poly(\log n)$. We first discuss the recourse bound.
To this end, we relate the adaptive 2-split game to an intermediate (adaptive) game, denoted $\ppay$, where: (1) balls from deleted bins are thrown independently uniformly at random, instead of using the 2-split strategy; and (2)  
when a bin with $k$ balls is deleted, the incurred recourse is $k+O(k\log k)$ instead of $k$. Note that the first modification ``helps us'' in that it makes the problem easier to analyze and the second modification ``helps the adversary'' as the recourse is increased.

Our analysis proceeds using two main claims: (1) The recourse in $\ppay$ is at total variation distance $\tfrac{1}{n^8}$ from a random variable that stochastically dominates the recourse in $2$-split; and (2) the recourse in $\ppay$ is $O(n)$ with high probability.
By a union bound, the recourse in $2$-split is $O(n)$ with high probability.
The proof of the second claim is very similar to the uniform game, and is omitted from this overview.
The basic idea behind the first claim is that if we consider $k\geq 2$ balls that move because of a single deletion, and track their movement until each of the $k$ balls become isolated from the other $k-1$ balls, then with high probability, they do not inflict a lot of recourse (concretely, $O(k\log k)$ recourse). 
Therefore, we can basically replace each deletion by paying this `cost', and throw the balls independently.
This is formulated using the following toy problem and an inductive argument. In the toy problem, there are $2\leq k\leq O(\tfrac{\log n}{\log\log n})$ balls, initially located at the same bin. In each round, the adversary deletes a bin that contains at least $2$ balls, and the balls are split evenly to two sets that are thrown independently. The game ends when the adversary cannot delete bins. In \Cref{sec:2_split_recourse}, we show that with high probability, the recourse in the toy problem is $O(k\log k)$, and complete the proof using an inductive argument.

\paragraph*{$2$-split, maximum load.}
To show that after $\tfrac{n}{2}$ rounds the maximum load is $\poly(\log n)$, we consider the following exponential potential function. For a load vector $(L_1,\ldots,L_n)$ and a set $I\subseteq [n]$ denoting the set of surviving bins, the potential is 
\[
w=\sum_{i\in I} e^{\alpha L_i},
\]
where $\alpha = \tfrac{1}{\log n}$.
By tracking $w$, we obtain that there exists a constant $C>1$, such that with probability at least $1/2$, when the maximum load increases by a $\log n$ factor, at least $n/C$ rounds passed. We repeatedly apply such an argument, and obtain that every time the maximum load increases by a $\log n$ factor, with probability $1/2$, at least $n/C$ rounds passed.
The $\poly(\log n)$ bound on the maximum load after $n/2$ rounds then follows by standard arguments.

 \section{Preliminaries}
Our proofs use the following concentration bounds.
\begin{lemma}[Hoeffding's inequality]\label{lem:Hoeffding}
    Let $X_1,\ldots,X_n$ be independent random variables, such that $0\leq X_i\leq C$ for all $i$.
    Then,
    \[
    \Pr\Big(\Big|\sum_{i=1}^n X_i-\sum_{i=1}^n\E X_i\Big|>t\Big)\leq 2\exp\Big(-\frac{2t^2}{nC^2}\Big).
    \]
\end{lemma}
\begin{lemma}[Bernstein's inequality]\label{lem:Bernstein}
    Let $X_1,\ldots,X_n$ be independent random variables, such that $0\leq X_i\leq C$ for all $i$.
    Denote $S=\sum_{i=1}^n X_i$ and $V=\var(S)$.
    Then,
    \[
    \Pr(|S-\E S|>t)\leq 2\exp\Big(-\frac{t^2/2}{V+Ct/3}\Big).
    \]
\end{lemma}
 \section{Uniform and two choices allocations}\label{sec:uniform}
Recall that we consider a $2$-player game between an adaptive adversary and an algorithm.
Initially, there are $n$ balls independently distributed uniformly at random over $n$ bins.
The adversary sees the assignment of balls into bins, and picks a bin to delete.
The algorithm must distribute the balls from the deleted bin among the remaining bins. 
This game continues in this manner for $T$ rounds.
The algorithm has two objectives: to minimize the number of balls that it distributes, called recourse, and to minimize the maximum load.
We first consider the algorithm that independently distributes the balls from the deleted bin uniformly at random among the remaining bins (for brevity, we simply say the balls are thrown to the remaining bins), which we call the uniform game.
Previously, \cite{FineKS25} considered this game for $T=n-1$, i.e., the game continues until only one bin remains. They showed that the recourse is $O(n\log n)$, which is the best possible for any algorithm.
We consider this game for intermediate number of rounds, and show the following.

\begin{theorem}\label{thm:adaptive_balls_bins_arbitrary_rounds}\label{thm:adaptive_balls_bins_n/2_rounds} Let $T\in [n-1]$, and denote $\ell=\lceil \log\tfrac{n}{n-T} \rceil$.
    In $T$ rounds of the uniform game, with high probability,
    the recourse is $O(\ell\cdot n)$ and the maximum load is 
\begin{align*}
        &\begin{cases}
        O(\log n \log \tfrac{\log\log n}{\log\log n -\ell})
& \text{if $\ell<\log\log n$ (corresponding to the lightly-loaded case),} \\
        O(2^\ell+\log n\log\log\log n) & \text{if $\ell\geq \log\log n$, (corresponding to the heavily-loaded case).}
    \end{cases}\end{align*}
\end{theorem}
\begin{remark}
    Put differently, the maximum load is 
    \[
    \begin{cases}
        O(\tfrac{\ell\log n}{\log\log n}) &\text{if $\ell\leq \tfrac{1}{2}\log\log n$,} \\
        O(\log n \log \tfrac{\log\log n}{\log\log n -\ell})
        & \text{if $\tfrac{1}{2}\log\log n<\ell<\log\log n$,} \\
        O(\log n \log \log\log n )
        & \text{if $\log\log n \leq \ell<\log\log n+\log\log\log\log n$,} \\
        O(2^\ell) & \text{if $\ell\geq \log\log n+\log\log\log\log n$.}
    \end{cases}
    \]
    That is, for small $\ell$ it interpolates from the standard bound $O(\tfrac{\log n}{\log\log n})$ of the lightly-loaded case to $O(\log n \log \log\log n )$, and for large $\ell\geq \log\log n+\log\log\log\log n$, the maximum load is the same order as the average load. Recall that when throwing $m'=n$ balls into $n'=n/2^\ell$ bins uniformly at random, by \cite{RaabS98}, the maximum load is 
\[
    \begin{cases}
        O\Big(\tfrac{\log n}{\log\frac{\log n}{2^\ell}}\Big) & \text{if $\ell<\log\log n$ (light load),} \\
        2^\ell+O(\sqrt{2^\ell \log n}) & \text{if $\ell\geq \log\log n$ (heavy load),}
    \end{cases}
    \]
hence our bound for the maximum load is nearly tight, for all regimes.
\end{remark}

The recourse bound is the best possible, for any algorithm, as explained below, and matches the bound of \cite{FineKS25} for $T=n-1$.

The proof of \Cref{thm:adaptive_balls_bins_arbitrary_rounds} considers $\ell$ phases. The $j$-th phase, where $j\in [\ell]$, is all rounds numbered in $[n-\tfrac{n}{2^{j-1}}+1,n-\tfrac{n}{2^{j}}]$, i.e., the first phase is the first $n/2$ rounds, the second phase is the next $n/4$ rounds, etc. Phase $j$ consists of exactly $n/2^j$ rounds and begins with $n/2^{j-1}$ alive bins.
Consider the balls that are moved in the $j$-th phase, and denote by $X_j$ their recourse, and by $A_j$ the maximum load they incur.
Clearly, the total recourse is $X=\sum_{j=1}^\ell X_j$, and the maximum load is $A\leq \sum_{j=1}^\ell A_j$. We show that with high probability, for all $j\in[\ell]$,
\begin{itemize}
    \item $X_j= O(n)$, and
    \item $A_j = 
    \begin{cases}
        O\Big(\tfrac{\log n}{\log\frac{\log n}{2^j}}\Big) & \text{if $j<\log\log n$ (light load),} \\
        O(2^j) & \text{if $j\geq \log\log n$ (heavy load).}
    \end{cases}
    $
\end{itemize}
The recourse bound is the best possible, since the average load in the $j$-th phase is $\Theta(2^j)$, and there are $\tfrac{n}{2^j}$ rounds. The maximum load matches that of throwing $O(n)$ balls to $\tfrac{n}{2^j}$ bins, and thus is the best possible.
The proof follows by summation over the $\ell$ phases, which might be slightly sub-optimal for bounding the maximum load when $\omega(1)\leq \ell\leq \log\log n+\log\log\log\log n$.

The proof uses
the following two lemmas. 
The first reduces (via coupling) each phase $j$  to the following non-adaptive game, and the second bounds the recourse and maximum load in this non-adaptive game.
In the $j$-th non-adaptive game, balls are thrown one at a time to $n/2^{j-1}$ bins until every subset of $n/2^{j}$ bins contains at least $n$ balls. We denote by $Y_j$ the number of inserted balls, and by $B_j$ the maximum load (which are random variables). 

\begin{lemma}\label{lem:simulation_uniform_j_game}
    Let $j\in [\log n]$.
  Fix an adaptive adversary, and denote its recourse in the $j$-th phase by $X_j$ and the load incurred in the $j$-th phase by $A_j$. Then, $Y_j\succeq X_j$ and $B_j\succeq A_j$, where $Y_j$ and $B_j$ are the number of inserted balls and the maximum load in the $j$-th non-adaptive game defined above.
\end{lemma}

\begin{proof}

We begin by analyzing the recourse. We construct a coupling between $Y_j$ and $X_j$ by simulating the non-adaptive game during the $n/2^j$ rounds of the adaptive game. Abusing notation slightly, we continue to use $X_j$ and $Y_j$ to denote these random variables within the coupled setting.

For the non-adaptive game, let $L^{na}_i$ denote the load of the $i$-th bin for $i\in [n/2^{j-1}]$. In the adaptive game, denote by $\alive_j\subseteq [n]$ the set of living bins right before the $j$-th phase starts. In particular, $|\alive_j|=n/2^{j-1}$, and therefore there is a one-to-one mapping between $\alive_j$ and $[n/2^{j-1}]$. For simplicity of presentation, we omit this mapping and simply treat every $i\in \alive_j$ as a number in $[n/2^{j-1}]$.
Let $\alive\subseteq \alive_j$ denote the set of living bins during phase $j$; for each $i\in \alive$, let $L^a_i$ be the load of the $i$-th bin incurred by balls that move in the $j$-th phase. 
We maintain the following invariant:

    \begin{itemize}
        \item \textbf{Invariant.} For every $i\in \alive$, we have $L^{na}_i=L^a_i$.
    \end{itemize}
    Initially, in the non-adaptive game there are no balls. No balls moved yet in the adaptive game. Clearly, the invariant holds.
    
    Now, consider a round of the adaptive game.
    The adaptive adversary picks a bin to delete, say containing $b$ balls.
    Suppose these $b$ balls are thrown one at a time in the adaptive game, and let us simulate throwing one ball using the non-adaptive game.
    In the non-adaptive game, we throw a ball uniformly over all the $n/2^{j-1}$ bins. If the ball is inserted to a bin in $[n/2^{j-1}]\setminus \alive$, i.e., bin that was deleted in the adaptive game, we repeatedly throw another new ball to the $n/2^{j-1}$ bins, until we get a ball in some bin $i\in \alive$. Therefore, for all $i'\in \alive, i'\neq i$, we have that $L^{na}_{i'}$ is unchanged, and $L^{na}_i\gets L^{na}_i+1$.
    We copy that last ball to the adaptive game, i.e., we only update $L^a_i\gets L^a_i+1$, and so the invariant holds.
    This process is repeated for all the $b$ balls, and then we proceed to the next round of the adaptive game.

    Crucially, we claim that the above is a simulation of the adaptive game, and it may end only after the adaptive adversary deletes $n/{2^j}$ bins.
    In the non-adaptive game, the distribution of a thrown ball, given that it lands in $\alive$, is uniform over $\alive$. Thus, the above indeed simulates the adaptive game.

    Observe that right before the adaptive adversary deletes the $\tfrac{n}{2^{j}}$-th bin, we have that $|\alive|=\tfrac{n}{2^j}+1$. Since $\sum_{i\in \alive} L^{a}_i\leq n$, there must exists a subset of $\alive$ of size $\tfrac{n}{2^j}$, whose bins in the adaptive game contain strictly less than $n$ balls.
    Therefore, by the invariant, also in the non-adaptive game, there is a subset of $\tfrac{n}{2^j}$ bins containing strictly less than $n$ balls, and thus the non-adaptive game didn't end.
    Hence, the number $Y_j$ of balls in the non-adaptive game is at least the recourse $X_j$ of the adaptive game in the $j$-th phase, i.e., $Y_j\geq X_j$.

    Over the same probability space, we immediately obtain a similar result for the maximum load.
    The maximum load $B_j$ in the non-adaptive game is larger than the maximum load in $\alive$, which by the invariant, equals the maximum load incurred by balls that move in the $j$-th phase of the adaptive game, and thus $B_j\geq A_j$. The proof is concluded by \Cref{lem:stoch_dominat_equivalent_coupling}.
\end{proof}

\begin{lemma}\label{lem:non_adaptive_game_length}
Consider the $j$-th non-adaptive game and let $Y_j$ be the number of balls it throws. Then,
    $Y_j\leq 10n$ with probability at least $1-e^{-2n}$.
\end{lemma}

\begin{proof}
    Consider the classical balls into bins problem with  $10n$ balls, and each ball is thrown uniformly to $n/2^{j-1}$ bins.
    We show that with high probability, all subsets of $n/2^{j}$ bins contain at least $n$ balls, and the claim follows.
    
    Let  $S\subset [n/2^{j-1}]$ be a subset of size $|S|=\tfrac{n}{2^{j}}$. For every $i\in[10n]$, denote by $I_i$ an indicator that the $i$-th ball is in $S$. 
    Thus, the number of balls in $S$ equals $\sum_{i=1}^{10n} I_i$. Clearly, $I_i=1$ w.p.\ $\tfrac{1}{2}$ and $0$ otherwise, and all these indicators are independent.
    Thus, by Hoeffding's inequality,
    \[
    \Pr\Big(\sum_{i=1}^{10n} I_i<n\Big) 
    \leq \Pr\Big(|\sum_{i=1}^{10n} I_i-5n|>4n\Big)
    \leq 2\exp\Big(-\frac{2(4n)^2}{10 n}\Big) \leq 2\exp(-3n).
    \]
    By a union bound, the probability that there exists a set $S\subset [n]$ of size $|S|=\tfrac{n}{2^{j}}$ containing less than $n$ balls is at most
    \[
    {n/2^{j-1}\choose n/2^{j}}\cdot 2\exp(-3n)
    \leq 2^{n}\cdot 2\exp(-3n)
    \leq \exp(-2n),
    \]
    which concludes the proof of the lemma.
\end{proof}

\begin{proof}[Proof of \Cref{thm:adaptive_balls_bins_arbitrary_rounds}.]
    By combining
    Lemma \ref{lem:simulation_uniform_j_game} and Lemma \ref{lem:non_adaptive_game_length},
 we have that for each $j\in [\ell]$,
    with probability at least $1-e^{-2n}$, the recourse is at most $10n$.
    By a union bound, with probability at least $1-e^{-n}$, the total recourse is $O(\ell n)$.
In standard balls and bins with $10n$ balls and $n/2^{j-1}$ bins, 
    with probability $1-\tfrac{1}{n^{11}}$, the maximum load is \cite{RaabS98}
    \[
    \begin{cases}
        O\Big(\tfrac{\log n}{\log\frac{\log n}{2^j}}\Big) & \text{if $j<\log\log n$,} \\
        O(2^j) & \text{if $j\geq \log\log n$.}
    \end{cases}
    \]
    Thus, by a union bound, with probability $1-n^{-10}$, the maximum load $B_j$
    in the non-adaptive game is bounded by that as well.
Observe that
    \begin{align*}
        \sum_{j=1}^\ell \tfrac{1}{\log\frac{\log n}{2^j}} &= \sum_{x=\log\log n -\ell}^{\log\log n-1} \tfrac{1}{x} \\
        &\leq \int_{\log\log n -\ell-1}^{\log\log n} \tfrac{dx}{x} = O(\log \tfrac{\log\log n}{\log\log n -\ell}).
    \end{align*}
    When $\ell\leq\tfrac{1}{2}\log\log n$, we have that $\log \tfrac{\log\log n}{\log\log n -\ell} = \log (1+\tfrac{\ell}{\log\log n -\ell})=O(\tfrac{\ell}{\log\log n})$.
    Thus, with probability $1-n^{-10}$, the maximum load
    in the adaptive game is 
    \begin{align*}
        &\leq \sum_{j=1}^\ell A_j \leq \sum_{j=1}^\ell B_j\\
        &\leq 
        \begin{cases}
        O(\log n \log \tfrac{\log\log n}{\log\log n -\ell})
& \text{if $\ell<\log\log n$,} \\
        O(2^\ell+\log n\log\log\log n) & \text{if $\ell\geq \log\log n$,}
    \end{cases}\\
    &= \begin{cases}
        O(\tfrac{\ell\log n}{\log\log n}) &\text{if $\ell\leq \tfrac{1}{2}\log\log n$,} \\
        O(\log n \log \tfrac{\log\log n}{\log\log n -\ell})
        & \text{if $\tfrac{1}{2}\log\log n<\ell<\log\log n$,} \\
        O(\log n \log \log\log n )
        & \text{if $\log\log n \leq \ell<\log\log n+\log\log\log\log n$,} \\
        O(2^\ell) & \text{if $\ell\geq \log\log n+\log\log\log\log n$,}
    \end{cases}
    \end{align*}
and the proof follows.
\end{proof}

\subsection{Two choices}
Consider the following adaptation of the game above: whenever a ball is thrown, we pick two independent bins $b_1,b_2$, and place the ball at the least loaded of $b_1,b_2$. This is applied both to the initial configuration, and to balls thrown during the game.
This is based on the well-known power of two choices, which in standard (non-adaptive) balls and bins setting, yields a maximum load of $O(\log\log n)$ when $m=O(n)$~\cite{AzarBKU99,BCSV00_2choice_heavy,BCSV06_2choice_heavy-sicomp}.
We show the following.
\begin{theorem}\label{thm:adaptive_balls_bins_two_choice}
    Let $T\in [n-1]$ and denote $\ell=\lceil \log\tfrac{n}{n-T} \rceil$.
In $T$ rounds of adaptive two choices game, the recourse is $O(\ell n)$ and the maximum load is $O(2^\ell+\ell\log\log n)$, with high probability.
\end{theorem}
The proof is similar to that of \Cref{thm:adaptive_balls_bins_n/2_rounds}.
We reduce the problem (via coupling) to a non-adaptive game, and then bound the parameters of this non-adaptive game.
The non-adaptive game is similar, only now we consider pairs of indices. For simplicity, consider $n/2$ rounds (i.e., $\ell=1$). In this new non-adaptive game, iteratively pick two indices independently and uniformly at random over $[n]$, place a ball in the least loaded of them, and stop when every subset of size $n/2$ contains at least $n$ pairs. Denote by $Y$ the number of inserted balls, and by $B$ the maximum load (which are random variables). 
We show in \Cref{appendix_two_choice} that \Cref{lem:simulation_uniform_j_game} still holds in this setting, using the same coupling,
and the number of rounds here is still $O(n)$ with high probability.
Thus, with high probability, the recourse is $O(n)$.
To bound the maximum load, recall that when throwing $m$ balls into $n$ bins using two choices allocations, if $m\geq n$, the maximum load is at most $\tfrac{m}{n}+O(\log\log n)$ with high probability~\cite{AzarBKU99,BCSV00_2choice_heavy}.
In our case, $m=10n$, thus the maximum load is $O(\log\log n)$, yielding \Cref{thm:adaptive_balls_bins_two_choice}.
A proof is provided in \Cref{appendix_two_choice}.

\section{$d$-split problem, $d=2$}\label{sec:2_split}
We now consider the $2$-split problem, played for $n/2$ steps.
That is, when a bin is deleted, the balls inside are partitioned into two uniformly random (remaining) bins. We show that in $\tfrac{n}{2}$ rounds, the recourse is $O(n)$ with high probability, and the maximum load is $\poly(\log n)$ with constant probability. While the constant $2$ in the number of rounds is arbitrary, bounding the recourse and maximum load when the number of rounds approaches $n$ remains open.
\subsection{Bounding the recourse}\label{sec:2_split_recourse}
\begin{theorem}\label{thm:2_split_adaptive}
    In $\tfrac{n}{2}$ rounds of the $2$-split game,
    the recourse is $O(n)$ with probability $1-\tfrac{1}{n^7}$.
\end{theorem}
As we mentioned in the introduction, our proof uses an adaptive balls into bins game, denoted $\ppay$, where there is an extra $O(k\log k)$ cost for moving $k\geq 2$ balls from a bin with $k$ balls, and the balls are thrown independently uniformly at random.
\begin{lemma}\label{lem:payment_for_splits_cost}
    In $n/2$ rounds of $\ppay$, the total cost is $O(n)$ with probability $1-\tfrac{1}{n^8}$.
\end{lemma}
\begin{lemma}\label{lem:reduction_2_split_to_payment}
    The total cost in $n/2$ rounds of $\ppay$ stochastically dominates a RV that is at total variation distance at most $\tfrac{1}{n^8}$ from the recourse in $n/2$ rounds of the $2$-split game.
\end{lemma}
\Cref{thm:2_split_adaptive} follows from these two lemmas and a union bound. We first prove \Cref{lem:reduction_2_split_to_payment}. The proof of \Cref{lem:payment_for_splits_cost} is provided in \Cref{sec:proof_ppay}.

\subsubsection{Proof of \Cref{lem:reduction_2_split_to_payment}}
Recall that our proof uses a toy problem, as follows.
There are $k\geq 2$ balls, initially located at the same bin. In rounds, the adversary deletes a bin that contains at least $2$ balls, and the balls are split evenly to two sets that are thrown independently. The game ends when the adversary cannot delete bins, that is, every bin has at most one ball.
\begin{lemma}\label{lem:recourse_toy}
    If $k= O(\log n)$, then with probability $1-\tfrac{1}{n^{10}}$, the recourse in the toy problem is $O(k\log k)$.
\end{lemma}
The proof of \Cref{lem:recourse_toy} is given in \Cref{sec:proof_toy_problem}.
We now prove \Cref{lem:reduction_2_split_to_payment} using \Cref{lem:recourse_toy}.
\begin{proof}[Proof of \Cref{lem:reduction_2_split_to_payment}.]

To connect $2$-split to $\ppay$, both played for $n/2$ rounds, we consider a sequence of hybrid games, 
$\pi_0,\pi_1,\dots,\pi_{n/2}$, where:
\begin{enumerate}
    \item[(1)] Game $\pi_0$ is exactly the $2$-split game.
    \item[(2)] For every $t\in[n/2]$, with high probability, it holds that the cost of $\pi_t$ is at least the cost of $\pi_{t-1}$.
    \item[(3)] The cost of $\ppay$ stochastically dominates the cost of $\pi_{n/2}$.
\end{enumerate}

Recall that in the game $\ppay$, whenever a bin with $k$ balls is deleted we incur an extra cost of $O(k\log k)$, and the balls are thrown independently uniformly at random. The sequence of hybrid games we consider $(\pi_0,\pi_1,\dots,\pi_{n/2})$ gradually interpolate between the $2$-split game and this $\ppay$ game. Intuitively, but somewhat inaccurately, in the game $\pi_t$, the \underline{first $t$ deletions} behave like in $\ppay$ (balls are thrown independently and our recourse has a logarithmic blowup), and the rest of the $\frac{n}{2}-t$ deletions behave like in the 2-split game (balls are partitioned into two groups). 
We now make this formal, starting with the formalization of the game $\pi_1$.

\paragraph*{The game ${\boldsymbol{\pi_1}}$.}

The game $\pi_1$ is identical to the $2$-split game, with one key modification regarding the balls from the {\em first} bin that is deleted. Suppose that the first deletion is of a bin with $k$ balls. We pay an extra cost of $O(k\log k)$ and redistribute the $k$ balls {\em independently and uniformly at random}. From now on, these $k$ balls are special; we call these $k$ balls ``Teflon'' and treat them differently throughout the rest of the game $\pi_1$. Unlike normal balls, Teflon balls do not ``stick''. That is, throughout the rest of the game $\pi_1$, when a bin is deleted, if this bin contains Teflon balls, then these Teflon balls are thrown independently and uniformly at random into the remaining bins. Note that ``normal'' balls from the deleted bin are still split into two sets and thrown as groups. Teflon balls may eventually revert to being ``normal'' balls based on  coupling conditions we describe next between the executions of $\pi_1$ and the 2-split game.

\paragraph*{The coupling between 2-split and ${\boldsymbol{\pi_1}}$.}

 Fix an adversary $\AAA$ for playing in the 2-split game (i.e., in $\pi_0$). We consider a coupled execution of 2-split and $\pi_1$ defined as follows. The initial configuration is identical in both games, obtained by throwing $n$ balls uniformly into the $n$ bins. The bin deletion sequence is kept identical in both games, and is determined by $\AAA$ based on the configuration of the balls in the 2-split game ($\AAA$ does not see the configuration of the balls in $\pi_1$). The first deletion made by $\AAA$ defines the identity of the Teflon balls in $\pi_1$. We stress that in the $2$-split game, Teflon balls behave as any other ball. To emphasize this, whenever we refer to these balls in the $2$-split game we will call them the ``so-called Teflon balls''. 
 Throughout the execution, whenever a bin is deleted, we redistribute its balls as follows in each of the games:
 \begin{itemize}
     \item In both games, normal balls are split evenly and thrown to the same randomly chosen bin index.

     \item In $\pi_1$, Teflon balls are thrown uniformly and independently.

     \item In the $2$-split game, Teflon balls from the deleted bin are split evenly to two sets and attached to the two sets of the normal balls. The first time when the bin of so-called Teflon-ball $b$ is deleted in the $2$-split game, and $b$ is assigned to a set that includes no other so-called Teflon-balls, we declare $b$ as normal (not Teflon anymore in both games), and throw its set to the bin of the same index as the bin containing $b$ in $\pi_1$. That is,  the randomness used to pick the identity of the new bin is taken from $\pi_1$.
    Since $b$ was previously thrown independently in $\pi_1$, then from the viewpoint of the $2$-split game, the thrown set is distributed independently, as desired.
 \end{itemize}

    Let $X$ denote the total recourse that $\AAA$ achieves in $n/2$ rounds of the $2$-split game, and let $Y_1$ denote the incurred cost in the coupled execution of $\pi_1$. We show that with high probability, $Y_1\geq X$.

\begin{claim}\label{claim:2_split_reduction_basis}
        With probability $1-\tfrac{1}{n^{9}}$, $Y_1\geq X$.
\end{claim}
    \begin{proof}
        By the invariant, the recourse of the normal balls is the same. We ignore the contribution to $Y_1$ from the additional recourse the Teflon-balls incur in $\pi_1$.
        Thus, we only have to bound the recourse of the so-called Teflon balls in the $2$-split game until they become normal, and compare it to the extra $O(k\log k)$ cost paid in $\pi_1$ at the first round.
        By standard arguments, with probability $1-\tfrac{1}{n^{10}}$, the maximum load in the initial distribution of $n$ balls into $n$ bins is $O(\log n)$, 
        and thus $k\leq O(\log n)$. Observe that in the special case where $k=1$, the two games are identical, hence we consider $k\geq 2$.
        The recourse of the so-called Teflon-balls in the $2$-split game 
        is stochastically dominated by the recourse in the toy problem, which by \Cref{lem:recourse_toy} is $O(k\log k)$ with probability $1-\tfrac{1}{n^{10}}$. By a union bound, we have $Y_1\geq X$ with probability $1-\tfrac{2}{n^{10}}$, which concludes the proof.
    \end{proof}

\paragraph*{The sequence of coupled games.}

    We now inductively define the sequence of coupled games $\pi_2,\dots,\pi_{n/2}$, where each game $\pi_i$ introduces an additional set of Teflon balls, denoted $T_i$ (where $T_1$ is the set of Teflon balls w.r.t.\ $\pi_1$ and where $T_0=\emptyset$). For $t\geq 2$, we define the game $\pi_t$, and couple it with $\pi_{t-1}$, similar to the construction above. The two games $\pi_{t-1}$ and $\pi_t$ are played the same for $t-1$ rounds. In round $t$, we define $T_t$ to be the set of all balls from the deleted bin that are not in $\cup_{i=1}^{t-1}T_i$. That is, $T_t$ contains all normal balls (w.r.t.\ $\pi_{t-1}$) that are deleted in round $t$. Now, from round $t$ onward:

    \begin{itemize}
        \item In $\pi_{t-1}$, balls that are not in $\cup_{i=1}^{t-1}T_i$ (aka normal w.r.t.\ $\pi_{t-1}$) are played as in $2$-split, and balls in $\cup_{i=1}^{t-1}T_i$ (aka Teflon) are played as independent. Some Teflon-balls may become normal, depending on the coupling with $\pi_{i}$, for $i\leq t-2$.

        \item The only difference between $\pi_{t-1}$ and $\pi_t$ is about the balls in $T_t$, where in $\pi_t$ these are played as Teflon-balls, and have extra cost $O(|T_t|\log |T_t|)$.
        
        \item In $\pi_{t-1}$, whenever a bin containing $T_t$-balls is deleted, we split the $T_t$-balls in that bin between the two sets. When a $T_t$-ball $b$ is assigned to a set that contains no other $T_t$-ball in $\pi_{t-1}$, we throw its set to the bin of the same index as the bin containing $b$ in $\pi_t$, and remove $b$ from $T_t$.
        As before, from the viewpoint of $\pi_{t-1}$, this assignment is independent.
    \end{itemize}

    \begin{claim}
        Let $t\in [n/2]$.
        With probability $1-\tfrac{1}{n^9}$, $Y_t\geq Y_{t-1}$.
    \end{claim}
    \begin{proof}
        This is a generalization of \Cref{claim:2_split_reduction_basis}, and the proof is similar.
        The only difference between $Y_{t-1}$ and $Y_t$ is in the cost and recourse caused by the Teflon balls $T_t$. We ignore the contribution to $Y_t$ from the additional recourse the Teflon balls in $T_t$ incur in $\pi_t$.
        Thus, we only compare the recourse of the balls in $T_t$ in $\pi_{t-1}$ until they become normal, with the cost $O(|T_t|\log |T_t|)$ paid in $\pi_t$.
        
        First, observe that with probability $1-\tfrac{1}{n^{10}}$, we have $|T_t|\leq O(\log n)$, by the following.
        For every $i\leq t-1$, the balls in a deleted bin in $\pi_{t-1}$ are distributed independently.
        Therefore, the maximum load before round $t$, and hence $|T_t|$, is stochastically dominated by the maximum load in uniform adaptive balls into bins played for $n/2$ rounds.
        Hence, by \Cref{thm:adaptive_balls_bins_n/2_rounds},
        $|T_t|\leq O(\log n)$ with probability $1-\tfrac{1}{n^{10}}$.

        The recourse in $\pi_{t-1}$ of the $T_t$-balls until they become normal, is stochastically dominated by the toy problem, with $k=|T_t|$. By \Cref{lem:recourse_toy} and a union bound, with probability $1-\tfrac{1}{n^9}$, the recourse in $\pi_{t-1}$ of the $T_t$-balls until they become normal is at most $O(|T_t|\log |T_t|)$.
        Under this same event, we have $Y_t\geq Y_{t-1}$, concluding the proof of the claim.
    \end{proof}
    To conclude the proof of \Cref{lem:reduction_2_split_to_payment}, observe that the cost of $\ppay$ stochastically dominates the cost of $\pi_{n/2}$, since the payments are only larger. (Because in  $\pi_{n/2}$ we do not necessarily incur the logarithmic blowup in the cost on all the balls from a deleted bin.)
    By a union bound, we have that with probability $1-\tfrac{1}{n^8}$, the sequence $Y_t$ is non-decreasing, and thus the cost of $\ppay$ is at least that of $2$-split.
\end{proof}

\subsubsection{Proof of \Cref{lem:recourse_toy}, bounding the recourse in the toy problem}\label{sec:proof_toy_problem}

\begin{proof}[Proof of \Cref{lem:recourse_toy}]
    Recall that in the toy problem, the initial configuration is that we have $2\leq k\leq O(\log n)$ balls in one bin, and the other bins are empty. Then, in rounds, the adversary deletes a bin
containing at least 2 balls, and the balls are split evenly to two sets that are thrown independently. The game ends when all bins contain at most $1$ ball.
    
    Towards analyzing the recourse, let us consider a modified variant of this toy problem, which we call $\pi'$, where the adversary has the additional power to determine the initial configuration for the $k$ balls. Furthermore, whenever a collision happens (i.e., whenever we throw a ball into a non-empty bin), the adversary can override the current configuration and completely rearrange the $k$ balls in the living bins. Clearly, the cost in $\pi'$ stochastically dominates the recourse in the toy problem, so it suffices to bound the cost in $\pi'$. To this end, we express $\pi'$ as a sequence of ``mini-games'' defined as follows:

    \paragraph*{Mini-game.}
    In a mini-game, the adversary initially picks a configuration for the $k$ balls in the living bins. Then, in every round, the adversary deletes a bin that contains at least $2$ balls, and the balls are split evenly to two sets that are thrown independently.
    Now, if one of these sets is thrown to a non-empty bin (i.e., there is a collision), we conclude the current mini-game and start a new one.
    Otherwise, we play until all bins contain at most $1$ ball, at which point we conclude the current mini-game without starting a new one (this corresponds to the end of $\pi'$).

    \medskip
    
    It is easy to see that a mini-game takes at most $O(k)$ rounds, and the actual recourse in a mini-game is bounded by $O(k\log k)$. 
Thus, to bound the total cost in $\pi'$, it suffices to show that the number of mini-games is small with high probability. That is, we want to show that the probability that a mini-game ends (i.e., there is a collision) without $\pi'$ ending is small.
    
Recall we have at least $n/2$ bins. Consider a round in a mini-game where the $k$ balls are located in $i$ bins. The adversary now deletes one of these bins, and the probability that one of the two sets we throw collide with a non-empty bin is at most $\tfrac{2i}{n/2}=O(\tfrac{k}{n})$. As there are at most $O(k)$ rounds in the mini-game, no matter what is the distribution that the adversary picked at the start of the mini-game, the probability of a collision is at most $O(\tfrac{k^2}{n})$.

    The total cost in $\pi'$ is the number of mini-games with a collision, times the cost $O(k\log k)$ of a mini-game.
    The number of mini-games is stochastically dominated by a geometric random variable, with success when there are no collisions, i.e., it has parameter $p=1-O(\tfrac{k^2}{n})$.
    Thus, the probability that there are more than $20$ mini-games is at most $O(\tfrac{k^2}{n})^{20}\leq O(n^{-10})$.
    When this event does not happen, the total cost in $\pi'$ is $O(k\log k)$.
    This concludes the proof.
\end{proof}

\subsubsection{Proof of \Cref{lem:payment_for_splits_cost}}\label{sec:proof_ppay}
Recall that in $\ppay$, the cost of deleting a bin with $k$ balls is $k+O(k\log k)$, and the balls are thrown independently.
\Cref{lem:payment_for_splits_cost} states that the cost of $n/2$ rounds of $\ppay$ is $O(n)$ with high probability.
\begin{proof}[Proof of \Cref{lem:payment_for_splits_cost}]
    The proof uses the same probability space as in the proof of \Cref{thm:adaptive_balls_bins_n/2_rounds}. We consider a non-adaptive game, where balls are thrown into $n$ bins until every subset of $n/2$ bins contains at least $n$ balls.
    By the coupling in the proof of \Cref{thm:adaptive_balls_bins_n/2_rounds}, we have a matching between deleted bin $i$ in the adaptive game and bin $i$ in the non-adaptive game, such that $L^{na}_i\geq L^a_i$, i.e., the load in the non-adaptive game is at least as in the adaptive game.
Consider the $n/2$ heaviest loaded bins in the non-adaptive game, with loads $k_1,\ldots,k_{n/2}$. Take as cost $O\big(\sum_{i=1}^{n/2} k_i (1+\log k_i)\big)$.
    By the coupling, this sum stochastically dominates the cost in $\ppay$.
    By \Cref{lem:non_adaptive_game_length}, with high probability, the number of thrown balls in the non-adaptive game is at most $10n$.
We now bound $\sum_{i=1}^{n/2} k_i \log k_i$ after throwing $10n$ balls, using the following.
    \begin{lemma}\label{lem:sum_LlogL_classic_balls}
        Consider standard balls into bins, with $10n$ balls and $n$ bins. Denote the loads of the bins by $L_1,L_2,\ldots,L_n$.
        With high probability, $\sum_{i=1}^n L_i\log L_i = O(n)$.
    \end{lemma}
    We defer the proof of \Cref{lem:sum_LlogL_classic_balls} to Appendix~\ref{sec:LlogL}.
    By a union bound, after throwing $10n$ balls, the game ends and the cost is $O(n)$, which concludes the proof of \Cref{lem:payment_for_splits_cost}.
\end{proof}

\subsection{Bounding the maximum load}
\begin{theorem}\label{thm:2_split_adaptive_load}
    The maximum load in the $2$-split problem is $\poly(\log n)$ with constant probability.
    Moreover, for success probability $1-\delta$, the maximum load is at most $\log^{O(\log \tfrac{1}{\delta})} n$.
\end{theorem}
Our proof uses the following lemma.

\begin{lemma}\label{lem:potential}
    Consider any assignment of $n$ balls to $m\in [n/2,n]$ bins, with maximum load $L$.
    There is a constant $C>1$, such that with probability $1/2$, the number of rounds of $2$-split until the maximum load reaches $O(L\cdot \log n)$ is at least $\min\{\tfrac{n}{C},m-\tfrac{n}{2}\}$.
\end{lemma}
The proof of \Cref{lem:potential} uses the following simple claim.
\begin{claim}\label{claim:ineq_in_potential}
    Let $w,n>0$.
    For every $x\geq 0$, we have $(w-x^2)(1+\tfrac{4x}{n})\leq w(1+\tfrac{4w}{n^2})$.
\end{claim}
\begin{proof}
    The difference satisfies
    \begin{align*}
        &w(1+\tfrac{4w}{n^2}) - (w-x^2)(1+\tfrac{4x}{n})\\
        &= w+\tfrac{4w^2}{n^2}-w-\tfrac{4wx}{n}+x^2+\tfrac{4x^3}{n} \\
        &= \tfrac{4w^2}{n^2}-\tfrac{4wx}{n} +x^2 +\tfrac{4x^3}{n} \\
        &= (\tfrac{2w}{n}-x)^2+\tfrac{4x^3}{n} \geq 0.
    \end{align*}
\end{proof}
\begin{proof}[Proof of \Cref{lem:potential}.]
    Pick $\alpha = 1/L$.
    Denote the set of surviving bins at time $t\in [0,m-\tfrac{n}{2}]$ by $I_t\subseteq[m]$, and for every $i\in I_t$, denote
    the load of the $i$-th bin as $L_i^t$. When clear from context, we simply write $L_i$.
    Consider the potential function 
    \[
    w_t=\sum_{i\in I_t} e^{\alpha L_i}.
    \]
    Initially, we have
    \[
    w_0=\sum_{i\in I_0} e^{\alpha L_i}\leq en.
    \]

    Without loss of generality we may
    suppose that the adversary is deterministic. That is, in every step, given the current constellation $H$, the adversary deterministically chooses the next bin to eliminate (denoted as $r$). 
    The only randomness is in picking the indices of the $2$ bins to which we split $r$, denoted as $x$ and $y$. For simplicity, we assume the algorithm enforces $x\neq y$. So going from time $t\leq m-\tfrac{n}{2}$ to time $t+1$, we have that $w_t$ changes to $w_{t+1}$ as follows:

    \begin{align*}
    w_{t+1} &=  w_t - e^{\alpha L_r}   +  e^{\alpha L_x}\left(e^{\alpha \lceil L_r/2\rceil} - 1\right)   +    e^{\alpha L_y}\left(e^{\alpha \lfloor L_r/2\rfloor} - 1\right)\\
    &\leq  w_t -  e^{\alpha L_r}   +   \left(2e^{\alpha L_r/2} - 1\right)  \left[ e^{\alpha L_x} +   e^{\alpha L_y}\right].
    \end{align*}
    
    We want to compute the expectation of this given the history $H$ of the game until time $t$. The expectation is only over the choice of $x,y$.

    \begin{align*}
        \E[w_{t+1}|H] &\leq  w_t - e^{\alpha L_r}   +   \left(2e^{\alpha L_r/2} - 1\right) 2  \E[e^{\alpha L_x}]\\
&\leq  w_t -  e^{\alpha L_r}   +  \left(2e^{\alpha L_r/2} - 1\right) 2  \sum_{i\in I_{t+1}}\tfrac{1}{|I_{t+1}|}  e^{\alpha L_i}\\
&\underset{|I_{t+1}|\ge n/2}{\leq}  w_t -  e^{\alpha L_r}   +   \frac{2e^{\alpha L_r/2} - 1}{n/4}  \left[w_t -  e^{\alpha L_r}\right]\\
&=\left[w_t -  e^{\alpha L_r}\right]\left(1+\frac{2e^{\alpha L_r/2} - 1}{n/4}\right)\\
&\leq \left[w_t -  e^{\alpha L_r}\right]\left(1+\frac{2e^{\alpha L_r/2}}{n/4}\right)\\
&\leq w_t\left(1+\frac{16 w_t}{ n^2}\right),
    \end{align*}   
    The last inequality holds by \Cref{claim:ineq_in_potential}.
    
    Let $C$ be some large constant and  consider the alternative random variables with an appropriate random stopping time $\tau$ defined as follows:
    \[
    \tau = \max \{ t : w_t \leq C n \},
    \qquad\text{and}\qquad
    \hat{w}_t = w_{t\land \tau},
    \]
    where $t\land \tau=\min\{t, \tau\}$. 
    Therefore,
    \[
        \E[\hat{w}_{t+1}|H]
        \leq \hat{w}_t (1+\tfrac{16 \hat{w}_t}{ n^2})
        \leq \hat{w}_t (1+\tfrac{16 C}{ n}),
    \]
    and so, by the law of total expectation and since $w_0\leq en$,
    \[
    E[\hat{w}_{n/C}]\leq en \cdot \left(1+\frac{16C}{n}\right)^{n/C}\leq
    en\cdot e^{16}.
    \]
    By Markov's inequality,
    \[
    \Pr[\hat{w}_{n/C}>2e^{17}n]\leq \frac{1}{2}.
    \]
    Set $C> 2e^{17} $. Thus, with probability at least $1/2$ we have that $\hat{w}_{n/C}\leq Cn$ in which case the stopping condition
does not happen until step $n/C$,
    and so $w_{n/C}=\hat{w}_{n/C}\leq Cn$ and $\tau\geq n/C$.
    Observe that for every $t\leq \tau$, the maximum load is at most,
    \[
    \tfrac{1}{\alpha} \log w_t \leq L\log(Cn),
    \]
    which concludes the proof of \Cref{lem:potential}.
\end{proof}

Before proceeding to the proof of \Cref{thm:2_split_adaptive_load}, we recall useful probability statements regarding martingales (and particularly submartingales).
\begin{definition}
    A sequence $X_1,X_2,\ldots$ of random variables with finite mean, is called submartingale if $\E[X_{i+1}|X_{i},\ldots,X_1]\geq X_{i}$, for all $i\geq 1$; is called supermartingale if $\E[X_{i+1}|X_{i},\ldots,X_1]\leq X_{i}$, for all $i\geq 1$; and is called martingale if it is both submartingale and supermartingale.
\end{definition}
We will use the following (special case of) Azuma's inequality for submartingales.
\begin{theorem}[Azuma's inequality for submartingales]
    Suppose $X_1,X_2,\ldots$ is a submartingale and $|X_{i+1}-X_i|\leq c$. Then, for all $n\in \N$ and $t>0$,
    \[
    \Pr(X_n-X_1\leq -t)\leq \exp\Big(-\frac{t^2}{2nc^2}\Big).
    \]
\end{theorem}

\begin{proof}[Proof of \Cref{thm:2_split_adaptive_load}.]
    By standard arguments, the initial maximum load is $\le \log n$ with high probability (say $\ge 1-n^{-10}$ for concreteness). Assume henceforth that this event holds. 
    Denote by $\tau_i$ the number of rounds until the maximum load reaches $\log n \cdot \log^i Cn = O(\log^{i+1} n)$ where $C$ is a constant satisfying the requirements of Lemma \ref{lem:potential}. It follows that $\tau_0=0$.
    Define the index $t^*=\max_{\tau_i\leq n/2} i$, i.e., the index of the last $\tau_i$ before the $\tfrac{n}{2}$-th round of the game.
    Given the constellation at time $\tau_i$, for $i\leq t^*$, by \Cref{lem:potential}, $\tau_{i}-\tau_{i-1}\geq n/C$ w.p.\ at least $1/2$.
    Denote by $I_i$ the indicator that $\tau_{i}-\tau_{i-1}\geq n/C$, 
    and $S_t=\sum_{i=1}^{t\land t^*} I_i-\tfrac{t}{2}$.
    We have that 
    \begin{align*}
        \E [S_t \mid S_{t-1},\ldots,S_1]&=
        \begin{cases}
    			S_{t-1}+\E\left[I_t-\tfrac{1}{2} \mid S_{t-1},\ldots,S_1 \right], & \text{if $t\leq t^*$}\\
                S_{t-1}, & \text{otherwise}
    		 \end{cases}
\\
        &\geq S_{t-1}.
    \end{align*}
    Thus, $S_t$ is a submartingale.
    By Azuma's inequality,
    \[
    \Pr[S_t\leq -\tfrac{t}{4}]\leq \exp(-\frac{t}{32}).
    \]
    Assume this event does not hold.

    Now, assume by contradiction that $2C< t^*$. We obtain $\sum_{i=1}^{2C} I_i=S_{2C}+\tfrac{2C}{2}\geq \tfrac{2C}{4}=\tfrac{C}{2}$.
    Thus, at least $\tfrac{C}{2}$ of the indicators $\{I_i\}_{i=1}^{2C}$ are equal $1$, and thus $\tau_{2C}\geq \tfrac{C}{2}\cdot\tfrac{n}{C}=\tfrac{n}{2}$.
    Therefore, $2C\geq t^*$, contradiction.
    We thus obtained that $t^*\leq 2C$.
    Therefore, the maximum load is $O(\log^{2C+1}  n)$ with probability at least $1-\exp(-\frac{C}{16})-n^{-10}$. Setting $C=\max\{C'\log(1/\delta), 2e^5\}$ for an appropriate constant $C'$ concludes the proof.
\end{proof} \section{$d$-split problem, $d=1$, oblivious adversary}\label{sec:no_split}
We now consider the $d$-split problem, with $d=1$, i.e., all the balls in a deleted bin move to the same new bin; and play the game for $n/2$ steps. With adaptive adversary, this strategy is clearly bad, as the adversary can delete the heaviest bin and create a snowball effect. For example, in the initial configuration where every bin contains exactly one ball, the recourse is $\Omega(n^2)$ and the maximum load is $n/2$ (and similar bounds seem to hold for other initial configurations).
Nevertheless, we establish that this strategy is solid against oblivious adversaries.
Suppose the initial configuration has one ball in each bin. 
For simplicity of presentation, we call this variant oblivious $1$-split.

\begin{theorem}\label{thm:oblivious_no_split}
Consider oblivious $1$-split, i.e., where all balls in a deleted bin are thrown together to a single new bin, the adversary is oblivious, and the initial configuration where every bin contains exactly one ball.
    After $n/2$ steps,
    the total recourse is $O(n\log n)$ and the max-load is $O(\log n)$, with high probability.
\end{theorem}
Recall the intuition behind this result.
Consider a random \emph{forest} graph on $n$ vertices constructed as follows: for each step $i$ from $1$ to $n/2$, vertex $i$ uniformly samples a parent $j \in \{i+1, \dots, n\}$ and becomes a child of $j$. This process yields $n/2$ trees $T_1,\ldots,T_{\frac{n}{2}}$, whose roots are the vertices indexed $n/2+1$ through $n$.
As explained right after \Cref{thm:1_split_oblivious_intro}, 
the maximum load in oblivious $1$-split equals the size of the largest tree in this forest, and the recourse equals the sum of depths of all vertices in this forest.

\subsection{Bounding the recourse}\label{sec:no_split_recourse}
We first bound the recourse of oblivious $1$-split by $O(n\log n)$ with high probability.
Recall that $Y$ is a geometric RV with parameter $\tfrac{1}{2}$ if $\Pr(Y=j)=2^{-j}$ for every integer $j\geq 1$.
\begin{claim}[Lemma 5 of \cite{FineKS25}]
    For every $i\leq n/2$, the recourse of the ball that started in the $i$-th bin is stochastically dominated by a geometric random variable with parameter $\tfrac{1}{2}$.
\end{claim}
Therefore, for each $i\in [n]$, with probability $1-n^{-10}$, the recourse of the $i$-th ball is $O(\log n)$.
By a union bound, with probability $1-n^{-9}$, the recourse of the $i$-th ball is $O(\log n)$ for all $i\in [n/2]$.
Thus, the overall recourse is $O(n\log n)$ with probability $1-n^{-9}$.

\subsection{Bounding the max-load}\label{sec:no_split_max_load}
We now show that the size of each of the trees $T_1,\ldots,T_{\frac{n}{2}}$ is stochastically dominated by a Galton-Watson Process (aka branching process), which yields the $O(\log n)$ bound on the maximum tree size.
\begin{definition}[Galton-Watson Process]
    A stochastic process $\{X_i\}_{i\in \N}$ is called Galton-Watson (GW) if it follows a recurrence formula $X_1=1$ and $X_{i+1}=\sum_{j=1}^{X_i} Y_j^i$, where $Y_j^i$ are iid natural numbered RVs.
    The distribution $\mu$ of the $Y_j^i$ is called the offspring distribution.
    The total progeny of the process is $\sum_{i=1}^\infty X_i$.
\end{definition}

Consider a random variable that is a sum of $n/2$ independent Bernoulli random variables, with success probabilities $\tfrac{1}{n/2},\tfrac{1}{n/2+1},\ldots, \tfrac{1}{n-1}$, and denote its distribution by $\mu$.
\begin{lemma}\label{lem:tree_GW}
    Let $T$ be one of the trees defined above.
    The size of $T$ is stochastically dominated by the total progeny of a Galton-Watson Process $\{X_i\}_{i\in \N}$ with offspring distribution $\mu$.
\end{lemma}

\begin{proof}
    Without loss of generality, let $T$ be the random tree rooted at $n/2+1$, where every node $i \leq n/2$ chooses a parent $parent(i)$ uniformly and independently from $\{i+1, \dots, n\}$.
    We introduce a coupling of the size of $T$ to the total progeny of a GW Process $\{X_i\}_{i\in \N}$ with offspring distribution $\mu$.
    Observe that a GW process can be defined by a tree $T^{GW}$ s.t. the number of nodes in level $i$ is $X_i$, and the number of children of the $j$-th node in level $i$ is $Y_j^i$.
    The total progeny of the GW process is the size of $T^{GW}$.

    In our coupling, we enforce $T\subseteq T^{GW}$
by making each node in $T$ produce fewer or equal children than its counterpart in $T_{GW}$ under this coupling.
The number of children of the roots of $T$ and $T^{GW}$ are distributed according to $\mu$, so coupling them is trivial (we use the same random coins).

    Traverse $T$ from right to left, and
    denote the subtree we observed so far by $T_{now}$. That is, initially $T_{now}$ is just the root node $n/2+1$.
In an inductive manner, suppose $T_{now}$ is a subtree also of $T^{GW}$. We now add new nodes of $T$ to $T_{now}$, and match them with nodes of $T^{GW}$.
Denote by $i$ the leaf of $T_{now}$ with the maximum index.
    Every children $u$ of $i$ is chosen only from available indices 
(i.e., $u\notin T_{now}$),
    and must satisfy $u<i$.
    We denote the set of potential children of $i$ by $U_{now}$, i.e., $U_{now} = \{u<i : u\notin T_{now}\}$.
    By construction, the set of potential parents of a node $u\in U_{now}$ equals $[u+1,i]\cup ([i+1,n]\setminus T_{now})$, i.e., vertices in $[u+1,n]$ that are not internal nodes of $T_{now}$ (since $i$ is the leaf of $T_{now}$ with the maximum index).
Denote $P=[i,\tfrac{n}{2}]\setminus T_{now}$.

    Given $T_{now}$, the number of children of $i$ is a sum of $|U_{now}|$ independent Bernoulli random variables with success probabilities 
    \[\frac{1}{\tfrac{n}{2}+|P|+i-u}, \qquad u\in U_{now}.\] 
Recall that the offspring distribution $\mu$ in $T^{GW}$ is defined by a sum of independent Bernoulli's with success probabilities $\tfrac{1}{\tfrac{n}{2}+j}$ for $j\in [0,n/2-1]$.
    Observe that $|P|\leq \tfrac{n}{2}-i$. 
    Thus, we can match the $|U_{now}|$ Bernoulli variables describing $i$'s children to the corresponding Bernoulli variables in $T^{GW}$ that have identical success probabilities. This forces the number of children of $i$ in $T$ to be stochastically bounded by the number of children of $i$ in $T^{GW}$. The remaining Bernoulli variables in $T^{GW}$ simply utilize another independent source of randomness. 
    This concludes the inductive step.

    We thus obtained a coupling between $T$ and $T^{GW}$, s.t. $T\subseteq T^{GW}$. Therefore, $|T|\leq |T^{GW}|$, concluding the proof.
\end{proof}
We now bound the size of $T^{GW}$ (or equivalently, the total progeny of the GW process) using Bernstein's inequality (\Cref{lem:Bernstein}).
We first bound the expectation and variance of $Y\sim \mu$.
\begin{lemma}\label{lem:mu_expect_var_bound}
    Let $Y\sim \mu$. The following hold.
    \begin{itemize}
        \item $\E Y<0.7$, and
        \item $\var Y <0.7$.
    \end{itemize}
\end{lemma}
To prove \Cref{lem:mu_expect_var_bound}, we use the following claim.
\begin{claim}\label{claim:sum_harmonic}
   \[
   \sum_{i=1}^{n/2} \tfrac{1}{n-i}\leq \ln 2 + O(\tfrac{1}{n}).
   \]
\end{claim}
\begin{proof}
    Denote the $n$-th harmonic number by $H_n = \sum_{k=1}^n \frac{1}{k}$.
   Using the asymptotic expansion $H_n = \ln n + \gamma + \frac{1}{2n} + O(n^{-2})$ (where $\gamma$ is the Euler-Mascheroni constant), 
   \[
   H_{n-1} - H_{n/2-1} = (\ln n + \gamma) - (\ln(n/2) + \gamma) +O(\tfrac{1}{n}) = \ln 2+O(\tfrac{1}{n}).
   \]
\end{proof}
\begin{proof}[Proof of \Cref{lem:mu_expect_var_bound}.]
    We have,
$\E Y = \sum_{i=1}^{n/2} \tfrac{1}{n-i}$.
By \Cref{claim:sum_harmonic}, $\E Y\leq \ln 2 + O(\tfrac{1}{n})$.
    As $\ln 2\approx 0.693$, this concludes the proof of the first bullet.
    Denote by $I_i$ a Bernoulli random variable with success probability $\tfrac{1}{n+i}$, for $i=1,\ldots,n/2$.
    We have
    \[
    \var Y = \sum_{i=1}^{n/2} \var I_i\leq \sum_{i=1}^{n/2} \E I_i^2 = \sum_{i=1}^{n/2} \E I_i=\E Y<0.7,
    \]
    which concludes the proof.
\end{proof}

Let us now bound the total progeny of the GW process.
\begin{lemma}\label{lem:GW_tree_size}
    With probability $1-n^{-10}$, the total progeny of a GW process with offspring distribution $\mu$ is $O(\log n)$.
\end{lemma}
\begin{proof}
    Enumerate the nodes of the tree using a search order (e.g., Breadth-First Search). 
    Let $\xi_i$ be the number of children of the $i$-th node visited. 
    Since nodes reproduce independently, the sequence $\xi_1, \xi_2, \dots$ consists of i.i.d.\ random variables distributed according to $\mu$.
    Let $S_k$ be the number of unexplored nodes after visiting $k$ nodes.
    We have the following recurrence: Start with one root, $S_0 = 1$, and at each step $i$, remove the current node and add its children
    \[S_i = S_{i-1} - 1 + \xi_i.\]
    Thus, 
    \[S_k = 1 + \sum_{i=1}^k (\xi_i - 1).\]
    The process stops when there are no nodes to explore ($S_k = 0$). The total progeny $Z$ is exactly the smallest index $k$ such that $S_k = 0$. 
    Therefore, the event that the tree size is large ($Z \ge k$) implies that $S_i\geq 1$ for all $0\leq i\leq k$. 
    A loose but sufficient condition for $Z \ge k$ is that the sum of offspring is at least the number of parents processed (minus the initial root),
    \[Z \ge k \implies \sum_{i=1}^k \xi_i \ge k - 1.\]
    Recall that each $\xi_i$ is a sum of $\tfrac{n}{2}$ Bernoulli random variables (with different distributions), thus
    $\sum_{i=1}^k \xi_i$ is in fact a sum of $kn/2$ Bernoulli random variables.
    Hence, by \Cref{lem:mu_expect_var_bound} and Bernstein's bound,
    \[
    \Pr(\sum_{i=1}^k \xi_i \ge k - 1)\leq 2\exp(-\tfrac{(0.3 k)^2}{0.7k+0.3k/3}).
    \]
    Taking $k=O(\log n)$, we obtain that the total progeny $Z\leq k=O(\log n)$ w.p. $1-n^{-10}$, as desired.
\end{proof}
We are now ready to bound the max-load in \Cref{thm:oblivious_no_split}, and thus conclude its proof.
\begin{proof}[Proof of \Cref{thm:oblivious_no_split}, max-load.]
    By \Cref{lem:tree_GW}, the load of a bin is dominated by the size of a GW process with offspring distribution $\mu$. 
    By \Cref{lem:GW_tree_size}, its size is $O(\log n)$ with probability $1-n^{-10}$.
    A union bound over the $n/2$ trees concludes the proof.
\end{proof}

\section*{Acknowledgments}
\paragraph{Haim Kaplan:} Partially supported by the Israel Science Foundation (grant 1156/23), and the Blavatnik Research Foundation.

\paragraph{Uri Stemmer:} Partially supported by the Israel Science Foundation (grant 1419/24), and the Blavatnik Research Foundation.

\bibliographystyle{alphaurl}
\bibliography{references}

\appendix
\section{Uniform two choices game}\label{appendix_two_choice}
In this section, we prove \Cref{thm:adaptive_balls_bins_two_choice}, which states that if the algorithm uses the two choices allocation scheme, then after $n-n/2^\ell$ rounds of adaptive game, the recourse is $O(\ell n)$ and the maximum load is $O(2^\ell+\ell\log \log n)$.

The proof is similar to that of \Cref{thm:adaptive_balls_bins_arbitrary_rounds}. It uses the following non-adaptive games that use the two choices allocation scheme.
In the $j$-th non-adaptive game, we begin with $n/2^{j-1}$ empty bins. 
Then, we iteratively draw a pair of indices $(i,i')$ uniformly over $[n/2^{j-1}]^2$, and place a ball in the least loaded among bins $i,i'$. 
The game proceeds until every subset $S\subset[n/2^{j-1}]$ of size $|S|=n/2^{j}$ contains the two indices of at least $n$ pairs.
We denote by $Y_j$ the number of inserted balls, and denote by $B_j$ the maximum load of a bin at the end of the game.
\begin{lemma}\label{lem:simulation_two_choice}
    Let $j\in [\log n]$.
  Fix an adaptive adversary, and denote the recourse it obtains against two choices allocations in the $j$-th phase by $X_j$, and its maximum load in the $j$-th phase by $A_j$. Then, $Y_j\succeq X_j$ and $B_j\succeq A_j$, where $Y_j$ and $B_j$ are the number of inserted balls and the maximum load in the $j$-th non-adaptive game defined above.
\end{lemma}
\begin{proof}
The proof is similar to the single choice (uniform) game of \Cref{lem:simulation_uniform_j_game}.
    We first analyze the recourse.
    We create a coupling between $Y_j$ and $X_j$, by simulating the non-adaptive game during the $n/2^j$ rounds of the adaptive game.
Abusing notation slightly, we continue to use $X_j$ and $Y_j$ to denote these random variables within the coupled setting.

    For the non-adaptive game, let $L^{na}_i$ denote the load of the $i$-th bin for $i\in [n/2^{j-1}]$. In the adaptive game, denote by $\alive_j\subseteq [n]$ the set of living bins right before the $j$-th phase starts.
    In particular, $|\alive_j|=n/2^{j-1}$, and therefore there is a one-to-one mapping between $\alive_j$ and $[n/2^{j-1}]$. For simplicity of presentation, we omit this mapping and simply treat every $i\in \alive_j$ as a number in $[n/2^{j-1}]$.
    Let $\alive\subseteq \alive_j$ denote the set of living bins during phase $j$; for each $i\in \alive$, let $L^a_i$ be the load of the $i$-th bin incurred by balls that move in the $j$-th phase. 
We maintain the following invariant.
    \begin{itemize}
        \item \textbf{Invariant.} For every $i\in \alive$, we have $L^{na}_i\geq L^a_i$.
    \end{itemize}
    Initially, in the non-adaptive game there are no balls. No balls moved yet in the adaptive game. Clearly, the invariant holds.

    Now, consider a round of the adaptive game.
    The adaptive adversary picks a bin to delete, say containing $b$ balls.
    Suppose these $b$ balls are thrown one at a time in the adaptive game, and let us simulate throwing one ball using the non-adaptive game.
    In the non-adaptive game, to throw a ball, we independently draw two indices $i,i'\in [n/2^{j-1}]$, and allocate the ball to the bin with smaller load (i.e., bin $i$ if $L^{na}_i\leq L^{na}_{i'}$, and bin $i'$ otherwise). 
    If $(i,i')\notin \alive^2$, i.e., bin $i$ or $i'$ were deleted in the adaptive game, we repeatedly throw another new ball to the $n/2^{j-1}$ bins, until we get two indices in $(i,i')\in \alive$.
    The process of repeatedly throwing new balls may increase the load over $\alive$ in the non-adaptive game, but this doesn't damage the invariant.
    When obtaining two indices $i,i'\in \alive$, let us assume without loss of generality that $L^{na}_i\geq L^{na}_{i'}$, and so we allocate the ball in bin $i'$ in the non-adaptive game. 
    If $L^a_i\geq L^a_{i'}$ as well, we allocate the ball in bin $i'$ also in the adaptive game, and the invariant holds.
    Otherwise, $L^a_i<L^a_{i'}$, so we must allocate the ball in bin $i$ in the adaptive game.
Before allocating the ball, we have by the invariant, $L^a_i<L^a_{i'}\leq L^{na}_{i'}\leq L^{na}_i$, and since the loads are integers, $L^a_i\leq L^{na}_i-1$. 
    Thus, the invariant still holds after adding this one ball to bin $i$ in the adaptive game (which increases $L^a_i$ by $1$).
This process is repeated for all the $b$ balls, and then we proceed to the next round of the adaptive game.

    Crucially, we claim that the above is a simulation of the adaptive game, and it may end only after the adaptive adversary deletes $n/{2^j}$ bins.
    In the non-adaptive game, the distribution of a pair of indices $i,i'$, given that both land in $\alive$, is uniform over $\alive^2$. Thus, the above indeed simulates the adaptive game.

     Observe that right before the adaptive adversary deletes the $\tfrac{n}{2^{j}}$-th bin, we have that $|\alive|=\tfrac{n}{2^j}+1$. Since $\sum_{i\in \alive} L^{a}_i\leq n$, there must exists a subset $S$ of $\alive$ of size $\tfrac{n}{2^j}$, whose bins in the adaptive game contain strictly less than $n$ balls.
     Hence, in the adaptive game, there were strictly less than $n$ pairs of indices both in $S$. Therefore, by the construction, since $S\subset \alive$, there were strictly less than $n$ pairs of indices in $S$ also in the non-adaptive game, and thus the non-adaptive game didn't end.
    Hence, the number $Y_j$ of balls in the non-adaptive game is at least the recourse $X_j$ of the adaptive game in the $j$-th phase, i.e., $Y_j\geq X_j$.

    Over the same probability space, we immediately obtain a similar result for the maximum load.
    The maximum load $B_j$ in the non-adaptive game is larger than the maximum load in $\alive$, which by the invariant, is larger than the maximum load incurred by balls that move in the $j$-th phase of the adaptive game, and thus $B_j\geq A_j$. The proof is concluded by \Cref{lem:stoch_dominat_equivalent_coupling}.
\end{proof}

\begin{lemma}\label{lem:num_balls_two_choice_cover_all_n2}
    Consider the $j$-th two choices non-adaptive game and let $Y_j$ be the number of balls it throws. Then,
    $Y_j\leq 24n$ with probability at least $1-e^{-n}$.
\end{lemma}
\begin{proof}
    The proof is similar to the proof of \Cref{lem:non_adaptive_game_length}.
    Consider throwing $24n$ pairs of indices independently and uniformly over $[n/2^{j-1}]^2$.
    We show that with high probability, all subsets of $n/2^j$ bins contain at least $n$ pairs, and the claim follows.
    
    Let a set $S\subset [n/2^{j-1}]$ of size $|S|=\tfrac{n}{2^j}$. For every $i\in[24n]$, denote by $I_i$ an indicator that the $i$-th pair is in $S^2$. 
    Thus, the number of pairs in $S$ equals $\sum_{i=1}^{24n} I_i$. Clearly, $I_i=1$ w.p. $\tfrac{1}{4}$ and $0$ otherwise, and all these indicators are independent.
    Thus, by Hoeffding's inequality,
    \[
    \Pr\Big(\sum_{i=1}^{24n} I_i<n\Big) 
    \leq \Pr\Big(|\sum_{i=1}^{24n} I_i-6n|>5n\Big)
    \leq 2\exp\Big(-\frac{2(5n)^2}{24 n}\Big) \leq 2\exp(-2n).
    \]
    By a union bound, the probability that there exists a set $S\subset [n/2^{j-1}]$ of size $|S|=\tfrac{n}{2^j}$ containing less than $n$ balls is at most
    \[
    {n/2^{j-1}\choose n/2^j}\cdot 2\exp(-2n)
    \leq 2^{n}\cdot 2\exp(-2n)
    \leq \exp(-n),
    \]
    which concludes the proof of the lemma.
\end{proof}
\begin{proof}[Proof of \Cref{thm:adaptive_balls_bins_two_choice}.]
    By combining
    Lemma \ref{lem:simulation_two_choice} and Lemma \ref{lem:num_balls_two_choice_cover_all_n2},
 we have that for each $j\in [\ell]$,
    with probability at least $1-e^{-n}$, the recourse is at most $24n$.
    By a union bound, with probability at least $1-e^{-n/2}$, the total recourse is $O(\ell n)$.
    In standard two choices allocations with $10n$ balls and $n/2^{j-1}$ bins, by \cite{BCSV00_2choice_heavy},
    with probability $1-\tfrac{1}{n^{11}}$, the maximum load is 
    $B_j=O( 2^j + \log\log n )$.

    Thus, by a union bound, with probability $1-n^{-10}$, the maximum load
    in the non-adaptive game is 
    \[
    \leq \sum_{j=1}^\ell A_j \leq \sum_{j=1}^\ell B_j \leq O(2^\ell +\ell \log\log n),
    \]
and the proof follows.
\end{proof}

\section{Proof of \Cref{lem:sum_LlogL_classic_balls}}\label{sec:LlogL}
To prove \Cref{lem:sum_LlogL_classic_balls}, we apply McDiarmid's inequality.
\begin{lemma}[McDiarmid's inequality]
    Suppose a function $f:\R^n\to \R$ satisfies for all $i\in [n]$, and $x_1,\ldots,x_n\in\R$,
    \[
    \sup_{x'_i\in \R} \big|f(x_1,\ldots,x_i,\ldots,x_n)-f(x_1,\ldots,x'_i,\ldots,x_n)\big|\leq c.
    \]
    Let $X_1,\ldots,X_n\in \R$ be independent random variables.
    For all $t>0$,
    \[
    \Pr\big(|f(X_1,\ldots,X_n)-\E f(X_1,\ldots,X_n)|\geq t\big)\leq 2\exp\Big(-\frac{2t^2}{nc^2}\Big).
    \]
\end{lemma}
Recall that in \Cref{lem:sum_LlogL_classic_balls}, $10n$ balls are thrown to $n$ bins, the bins loads are denoted by $L_1,L_2,\ldots,L_n$, and the claim is that with high probability, $\sum_{i=1}^n L_i\log L_i = O(n)$.
Denote $F=\sum_{i=1}^{n} L_i\ln L_i$, which is $O(1)$ factor from the desired sum. We view $F$ as a function of the locations of the $10n$ balls.
Consider the effect of changing the location of one ball, say, by changing its location from bin $i$ to bin $j$. Then $F$ changes by
\begin{align*}
    &\big|(L_i-1)\ln (L_i-1) + (L_j+1)\ln (L_j+1) - L_i\ln L_i - L_j\ln L_j\big|    \\
    &= |L_i\ln (\tfrac{L_i-1}{L_i}) + L_j\ln (\tfrac{L_j+1}{L_j}) - \ln (L_i-1) + \ln (L_j+1)| \\
    &\leq 2\ln (10n) + |L_i\ln (\tfrac{L_i-1}{L_i})| + |L_j\ln (\tfrac{L_j+1}{L_j})| \\
    &\leq 2\ln (10n) + 2,
\end{align*}
where the last inequality is since $\ln(1+x)\leq x$ for $x>-1$.
\begin{claim}\label{claim:expected_sum_LlogL}
    We have $\E F \leq n\cdot 10\ln 11<25 n$.
\end{claim}
\begin{proof}[Proof of \Cref{lem:sum_LlogL_classic_balls}]
    By McDiarmid's inequality,
\begin{align*}
    \Pr(F > 100n)
    &\leq \Pr(F - \E(F) > 75n) \\
    &\leq 2\exp\Big(-\frac{2(75n)^2}{10n(2\ln (10n+2)+2)^2}\Big)\leq e^{-\sqrt{n}}.
\end{align*}
\end{proof}

\begin{proof}[Proof of \Cref{claim:expected_sum_LlogL}]
    By linearity of expectation, $\E F = n\cdot \E (L_1\ln L_1)$. 
    Observe that $L_1\sim Bin(10n,\tfrac{1}{n})$, i.e., it is a binomial random variable with $10n$ experiments and success probability $\tfrac{1}{n}$.
    Therefore,
    \begin{align*}
        \E L_1\ln L_1 &= \sum_{k=1}^{10n} {10n\choose k} \Big(1-\frac{1}{n}\Big)^{10n-k}\frac{1}{n^k}\cdot k\ln k \\
        &=\sum_{k=1}^{10n} \frac{10n}{k}{10n-1\choose k-1} \Big(1-\frac{1}{n}\Big)^{10n-k}\frac{1}{n^k}\cdot k\ln k && \text{since $k{m\choose k}=m{m-1\choose k-1}$} \\
        &= 10\sum_{k=1}^{10n} {10n-1\choose k-1} \Big(1-\frac{1}{n}\Big)^{10n-k}\frac{1}{n^{k-1}}\cdot \ln k \\
        &= 10\sum_{k=0}^{10n-1} {10n-1\choose k} \Big(1-\frac{1}{n}\Big)^{10n-1-k}\frac{1}{n^{k}}\cdot \ln (k+1) && \text{change of index} \\
        &= 10 \cdot \E \ln(X+1) && \text{where $X\sim Bin(10n-1,\tfrac{1}{n}$)} \\
        &\leq 10  \ln(\E X+1) && \text{by Jensen's inequality} \\
        &\leq 10 \ln 11<25.
    \end{align*}
\end{proof}

 \end{document}